\pdfoutput=1
\documentclass[a4paper,11pt]{article}
\usepackage{amsmath,amssymb,mathtools,amsthm}
\usepackage{hyperref}
\usepackage{cite}
\usepackage{graphicx}
\usepackage{authblk}
\usepackage{todonotes}
\usepackage{subfigure}
\usepackage{comment}
\usepackage{tikz}
\usepackage[margin=1in]{geometry}
\usepackage[T1]{fontenc}
\usepackage{color}
\usepackage{xhfill}

\usetikzlibrary{calc}
\usetikzlibrary{shapes.geometric, arrows,matrix}
\usetikzlibrary{decorations.markings}
\tikzset{nodeStyle/.style = {circle,draw,minimum size=30pt}}
\tikzset{arrowStyle/.style = {-latex}}

\newcommand{\ket}[1]{\ensuremath{|#1\rangle}}
\newcommand{\bra}[1]{\ensuremath{\langle#1|}}
\newcommand{\ketbra}[2]{\ensuremath{\ket{#1}\bra{#2}}}

\newcommand{\braket}[2]{\ensuremath{\langle{#1}|{#2}\rangle}}

\newcommand{\LL}{\mathcal{L}}
\newcommand{\MM}{\mathcal{M}}

\newcommand{\1}{{\rm 1\hspace{-0.9mm}l}}
\newcommand{\Id}{\1}
\newcommand{\ii}{\mathrm{i}}

\newcommand{\kron}{\otimes}
\newcommand{\con}[1]{\bar{#1}}
\newcommand{\vecc}[1]{ | #1 \rangle \rangle}
\newcommand{\Hl}{\mathcal H}

\newcommand{\C}{\mathbb C}

\newcommand{\ie}{\emph{i.e.\/}}

\newtheorem{theorem}{Theorem}
\newtheorem*{conjecture}{Conjecture}
\newtheorem{lemma}{Lemma}

\newtheorem{proposition}[theorem]{Proposition}
\newtheorem{remark}[theorem]{Remark}
\newtheorem*{property}{Property}

\title{Limiting properties of stochastic quantum walks on 
directed graphs}

\author[1,2]{Adam Glos}
\author[1]{Jaros\l aw Adam Miszczak}
\author[1,2]{Mateusz Ostaszewski\thanks{mostaszewski@iitis.pl}}
\affil[1]{Institute of Theoretical and Applied Informatics, Polish Academy of 
Sciences, \newline Ba{\l}tycka 5, 44-100 Gliwice, Poland}
\affil[2]{Institute of Informatics, Silesian University of Technology\newline
Akademicka 16, 44-100 Gliwice, Poland}

\date{}
\begin{document}
\maketitle

\begin{abstract}
The main results of our work is determining the differences between limiting
properties in various models of quantum stochastic walks. In particular, we prove that in the case of strongly connected and a class of weakly connected directed graphs, local environment interaction evolution is relaxing, and in
the case of undirected graphs, global environment interaction evolution is convergent.
For other classes of directed graphs we show, that the character of connectivity large influence on the limiting
properties. We also study the limiting properties for the non-moralizing global
interaction case. We demonstrate that the digraph observance is recovered in this case.
	
\end{abstract}


\section{Introduction \& preliminaries} \subsection{Motivation}
Results from past few decades show that the
choice of the quantum analogue of classical random walk is highly non-unique.
Some of the most popular models include discrete coined quantum walk
\cite{aharonov2001quantum}, continuous quantum walk \cite{childs2002example},
Szegedy walk \cite{szegedy2004quantum}, open walk \cite{attal2012open2},
staggered quantum walk \cite{portugal2016staggered}, and quantum stochastic walk
\cite{whitfield2010quantum}. They have found applications in designing PageRank
algorithm \cite{loke2017comparing,sanchez2011quantum,paparo2014quantum}, search
algorithms \cite{childs2004spatial,childs2003exponential,shenvi2003quantum, glos2017vertices},
solving triangle problem \cite{magniez2007quantum}, and describing chemical
reactions \cite{chia2016coherent}. All of them outperform their classical
counterpart at least for some large class of graphs. From the algorithmic
perspective, it is crucial to understand the limiting behaviour of the walks.
This includes the propagation speed, which decides about efficiency of the
search algorithm \cite{domino2017spontaneous}, mixing time and relaxing, which
find applications in PageRank algorithm
\cite{sanchez2011quantum,liu2016continuous,nayak2000quantum}, limit theorem
\cite{konno2005new,sadowski2016central}, and trapping \cite{sadowski2016lively}.

A special kind of continuous quantum evolution is quantum stochastic walk, which
generalizes both classical and quantum walk \cite{whitfield2010quantum}. The
model has been investigated in the context of relaxing property
\cite{sanchez2011quantum,liu2016continuous}, propagation speed
\cite{domino2016properties}, and applications to various physics
\cite{chia2016coherent} and computer science \cite{sanchez2011quantum} problems.
The difference comparing to original continuous-time quantum walk is the
Lindblad operators appearance. The choice of Lindblad operators is highly
nontrivial, and in the case were each one-dimensional subspace corresponds to
different vertex, two model are of particular interest: local environment
interaction and global environment
interaction~\cite{whitfield2010quantum,domino2016properties}.

Local environment interaction case has been extensively used and analysed. The model has been analysed in context of relaxation
\cite{liu2016continuous} and application in PageRank algorithm
\cite{sanchez2011quantum}. In particular, it was show that in the case of
undirected graphs, the QSW is always relaxing \cite{liu2016continuous}. However,
the evolution is decohering, and hence it destroys the ballistic propagation
\cite{bringuier2016central}.

The ballistic propagation was one of the basic motivation for the analysis of quantum walk. Fortunately, the global environment interaction QSW is proved to have
ballistic propagation~\cite{domino2016properties}. However, the global environment
interaction suffers from graph topology change~\cite{domino2017spontaneous}, due
to its character called \emph{moralization}. Thanks to the correction scheme it is possible
to bound the QSW with the global interaction to the digraph structure. Such
evolution, called non-moralizing global environment QSW, is verified numerically
to have at least superdiffusive propagation~\cite{domino2017spontaneous}.

The main contribution of this paper is the description of the limiting properties
of various models of quantum stochastic walks. We describe how the
connectivity of the graph influences the convergence or the relaxation. Numerical analysis suggests, at least
in some cases, that the convergence/relaxation appears on all graphs in the sense
of Erd\H{o}s-Renyi random graph model $\mathcal G(n,p)$. Furthermore our results includes the directed graph preservation analysis for directed graphs in the case of non-moralizing evolution.
The main results are collected in table~\ref{tab:results}

\begin{table}\centering
	\begin{tabular}{lp{3.7cm}p{3.7cm}p{3.7cm}}
		\hline
		& local interaction & global interaction & non-moralizing global interaction \\ \hline\hline
		convergence		 & -- unknown in general case&  -- undirected graphs \newline -- counterexample for digraphs  & -- counterexample for undirected graphs\\\hline
		relaxing  & -- strongly connected digraphs \newline
		-- weakly connected digraphs with one sink in condensation graph  & -- never for undirected graphs \newline -- counterexample for digraphs   &  -- counterexample for undirected graphs \\\hline
		
	\end{tabular}
	\caption{Main analytical results. Case $\omega\in (0,1]$ is only considered here.}
	\label{tab:results}
\end{table}

The paper is organized as follows. In the following part of this section we
provide some basic definitions concerning graph theory and QSW. In Sec.~\ref{sec:local} we analyse the limiting properties of the local
environment interaction case. In Sec.~\ref{sec:global} and
Sec.~\ref{sec:corrected-global} we analyse the limiting properties of the global
environment interaction case, the standard and the non-moralizing respectively. 
We conclude our
results in Sec.~\ref{sec:conclusions}.

\subsection{Graph theory terminology} Let $G=(V,E)$ be a digraph with vertex set
$V$ and arc set $E$. The underlying graph $G^u(V,E^u)$ is an undirected simple
graph, for which every arc from $G$ is replaced with an edge. We say that digraph is
weakly connected if its underlying graph is  connected. We say that digraph is
strongly connected, if for arbitrary $v,w\in V$ there is a directed path from
$v$ to $w$. The strongly connected components are the maximal strongly connected
subgraphs. We call $v$ a sink vertex, if its outdegree (number of outgoing arcs)
is zero. We denote $S(G)$ to be a set of sink vertices.

Let $\mathbb G = (\mathbb V, \mathbb E)$ be condensation of $G=(V,E)$, \ie{} a
directed graph constructed as follows. We make a partition $\mathbb V$ of vertex
set $V$ in such a way, that each block forms a maximal strongly connected
component in $G$. Then $(\mathbf v, \mathbf w)\in \mathbb E$ iff there exists
$v\in\mathbf v, w\in\mathbf w$ such that $(v,w)\in E$. In other words there is an
arc from one maximal strongly connected component to the other, if there is at
least one arc (in consistent direction) between their elements. Note that
$\mathbb G$ is directed acyclic graph, hence $S(\mathbb G)\neq\emptyset$.
Furthermore, if $G$ is weakly connected, then for each $\mathbf w\in \mathbb V$
there exists $\mathbf s\in  S(\mathbb G)$ such that there exists directed
path from $\mathbf w$ to $\mathbf s$.

The directed moral graph $G_M=(V,E_M)$ of the directed graph $G=(V,E)$ is defined as
follows. For each $v,v'\in V$ we have $(v,v')\in E_M$ iff $(v,v')\in E$ or there
exists $w\in V$ such that $(v,w),(v' ,w)\in E$. In other words, directed moral graph is
constructed by adding edge between the vertices which have common child. Note
that original moral graph is defined as underlying graph of the directed moral graph~\cite{Cowell1999, StudenyPhD}. 

In this work we are using Erd\H{o}s-R\'enyi model of random graphs.
The Erd\H{o}s-R\'enyi random graph $\mathcal{G}(n,p)$ is a graph with $n$
vertices such that any pair of vertices is connected with probability 
$p$~\cite{erdos1960evolution}.

Throughout this paper we will assume, that graph is at least weakly connected.

\subsection{GKSL master equation and quantum stochastic walks}  

To define quantum stochastic walks in general, let us start with the
Gorini-Kossakowski-Sudarshan-Lindblad (GKSL) master equation 
\cite{kossakowski1972quantum,
lindblad1976generators, gorini1976completely}
\begin{equation}
\frac{\textrm{d}}{\textrm{d}t}\varrho =\MM [\varrho]= -\ii [H, \varrho] + 
\sum_{L\in 
	\LL} \left(L \varrho 
L^\dagger - \frac12 \{L^\dagger L, \varrho\} \right), 
\label{eq:GKSL-master-equation}
\end{equation}
where $\{A, B\}$ is the anticommutator and $\MM$ is the evolution superoperator.
Here $H$ is the Hamiltonian, which describes the evolution of the closed system,
and $\LL$ is the collection of Lindblad operators, which describes the evolution
of the open system. This master equation describes general Markov continuous
evolution of mixed quantum states. Note, that in the case of $\LL=\emptyset$ and
$H$ being an adjacency matrix of some graph we recover the original continuous
quantum walk, however on mixed states.

The GKSL master equation was used for defining quantum stochastic walks (QSW),
which are a generalization of both classical random walks and quantum walks
\cite{whitfield2010quantum}. Both $H$ and $\LL$ correspond to the graph
structure, however one may verify that at least a choice of Lindblad operators
may be non-unique \cite{whitfield2010quantum, domino2016properties}. Suppose we
have a directed graph $G=(V,E)$. Since Hamiltonian $H$ needs to be hermitian, we
always choose adjacency matrix of the underlying graph graph. In the \emph{local
	environment interaction case} each Lindblad operator corresponds to a single
arc, $ \LL = \{c_{(v,w)}\ketbra{w}{v}\colon (v,w)\in E\}$ for $c_{(v,w)}\in
\C_{\neq 0}$. In the \emph{global environment interaction case} we choose a
single Lindblad operator, which is adjacency matrix of directed graph.

To analyse the impact of Lindbladian part, we add the smoothing parameter 
$\omega\in[0,1]$
\begin{equation}
\frac{\textrm{d}}{\textrm{d}t}\varrho =-\ii (1-\omega) [H, \varrho] + 
\omega \sum_{L\in 
	\LL} \left(L \varrho 
L^\dagger - \frac12 \{L^\dagger L, \varrho\} \right), 
\label{eq:GKSL-omega}
\end{equation}
In \cite{bringuier2016central} it was shown, by infinite path graph analysis,
that the local interaction case leads to the classical propagation of the walk.
Oppositely, in the case of the global interaction, the ballistic propagation is
obtained for arbitrary middle value $\omega$. Since for $\omega=0$ we recover
continuous quantum walk (CQW) on closed system, we are particularly interested
in $\omega \in(0,1]$ case.

In \cite{domino2017spontaneous} it has been notes that the global interaction QSW 
suffers for the graph topology change and the resulting
process fails to reproduce the structure of the original graph. 
The resulting graph, according to which the system is evolving, is the directed moral graph ~\cite{Cowell1999, StudenyPhD}, of the original one.
This effect is called a spontaneous moralization and to prevent this 
a correction scheme based on the system enlargement has been proposed~\cite{domino2017spontaneous}. The scheme consist of following steps:
first we combine with each vertex a subspace of the system of dimension
equal to the indegree of the vertex (in the case of source vertices the subspace
is onedimensional). Next, we choose a family of orthogonal matrices for
new Lindblad operator $\tilde L$, which destroys the spontaneous moralization.
The last step is to add a Hamiltonian $\tilde H_{\mathrm{local}}$ acting
locally on the subspaces corresponding to different vertices. The Hamiltonian
$\tilde H$ corresponding to the graphs structure is a 0-1 matrix, for which zero
values coincides with zero values of $\tilde L$. For details we refer the reader
to the original paper. Since in this model we combine each vertex with some
orthogonal subspaces, a natural measurement is the collection of operations
which projects the state onto the subspaces corresponding to the vertices.
We call this model of evolution as \emph{non-moralizing global environment 
interaction QSW}, or simply non-moralizing QSW. Similarly to Eq.~\eqref{eq:GKSL-omega}
we add smoothing parameter $\omega$ and the evolution takes the form
\begin{equation}
\frac{\textrm{d}}{\textrm{d}t}\varrho =-\ii (1-\omega) [\tilde H, \varrho] + 
\omega  \left(-\ii[\tilde H_{\mathrm{local}}, \varrho]+ \tilde L \varrho 
\tilde L^\dagger - \frac12 \{\tilde L^\dagger \tilde L, \varrho\} \right).  
\end{equation}

Numerical simulation of the QSW with non-moralizing global environment interaction is difficult because of enlarging of the system. With the increase of the input graph density, the size and density of output Lindblad operator $\tilde L$ increases rapidly.

Throughout this paper we analyse the limit behaviour of QSW in all of three
mentioned cases: local environment interaction, global environment interaction,
and corrected non-moralizing environment interaction. We analyse the evolution
in context of convergence and relaxation. We say that evolution is convergent, if
for arbitrary initial state $\varrho_0$ there exists stationary state
$\varrho_\infty$ such that $\varrho_t\xrightarrow{t\to\infty}\varrho_\infty$. We
say that evolution is relaxing, if there exists unique stationary state. In GKSL master equation uniqueness of stationary state is equivalent to
relaxing property, see Theorem 1 from \cite{schirmer_stabilizing_2010}.
Similarly one can define convergence and relaxing in the context of probability
distribution of quantum measurement. Note that relaxation implies convergence, but the opposite does not hold in general.

In the case of QSW $H$ and $\LL$ do not depend on time. Henceforth, we can solve
the differential equation analytically: if we choose initial state $\varrho_0$,
then
\begin{equation}
\vecc{\varrho_t} =\exp(tF) \vecc{\varrho_0}, \label{eq:evolution}
\end{equation}
where
\begin{equation}
F = -\ii \left(H \kron \Id - \Id \kron \con H \right) + 
\sum_{L\in\LL} \left ( L \kron \con L - \frac{1}{2}\left( L^\dagger L \kron 
\Id 
+ \Id \kron L^\dagger  L\right )\right 
),
\end{equation}
and $\vecc{\,\cdot\,}$ denotes the vectorization of the matrix (see eg.
\cite{miszczak2011singular}).
Note that the eigenvalues of $F$ implies the
behaviour of the evolution. If there exists purely imaginary nonzero
eigenvalues, the evolution is non-convergent for some initial state $\varrho_0$,
otherwise it is convergent. If the null-space is one-dimensional, then the
evolution is relaxing. Hence our numerical analysis is mostly based on analysing
the spectrum of $F$.

\section{Convergence of local interaction case} \label{sec:local}

The local environment interaction case is relaxing for undirected
graphs~\cite{liu2016continuous} and arbitrary Hamiltonian. The proof is based on
the Spohn theorem~\cite{Spohn1977}, which requires the self-adjointess of the set of Lindblad
operators, hence its applications is limited to the undirected graphs case.
Nevertheless we show, that the result can be extended to strongly connected
digraphs and weakly connected graphs with single sink vertex in $\mathbb G$ graph. Our proofs utilise the Condition 2. and Condition 3. from
\cite{schirmer_stabilizing_2010}, recalled here as Lemma~\ref{lemma:condition2} and \ref{lemma:condition3}. By interior we mean collection of density matrices with full rank.

\begin{lemma}[\cite{schirmer_stabilizing_2010}] \label{lemma:condition2}
Let $\Hl$ be a Hilbert space. If there is no proper subspace $S\subsetneq\Hl$, that is invariant under all Lindblad generators $L\in\LL$  then
the system has a unique steady state in the interior.
\end{lemma}

\begin{lemma}[\cite{schirmer_stabilizing_2010}] \label{lemma:condition3}
	If there do not exists two orthogonal proper subspaces of $\Hl$ that are simultaneously invariant under all Lindblad generators $L\in\LL$, then the system has unique fixed point, either at the boundary or in the interior.
\end{lemma}

Using the above lemmas we can prove the following.

\begin{theorem} \label{theorem:local-strongly}
Let $G=(V,E)$ be a strongly connected digraph and let $\LL=
\{L_{vw}=c_{(v,w)}\ketbra{w}{v}\colon (v,w)\in E\}$ for some $c_{(v,w)}\in \C_{\neq 0}$. Then
evolution described by Eq.~(\ref{eq:GKSL-master-equation}) is relaxing for
arbitrary Hamiltonian $H$ with stationary state in the
interior.
\end{theorem}
\begin{proof}

Let $\Hl$ be a
Hilbert space  spanned by $\{\ket v: v\in V\}$ and $S\neq \{0\}$ be arbitrary
subspace of $\Hl$ invariant under $\LL$. Furthermore, suppose that $\ket \psi\in S$ is a
nonzero vector and  $v\in V$ is such that $\braket{v}{\psi}\neq 0$.
Since $G$ is strongly connected, there is a directed path
$(v_1=v,v_2,\dots,v_k,w)$ for arbitrary $w\in V$. Then $c_w \ket w =L_{v_kw}
L_{v_{k-1}v_k}\cdots L_{v_2,v_1}\ket{\psi} \in S$ for some $c_w\in\C_{\neq
	0}$. Since $\{\ket{w} : w\in V \}$ forms a basis of $\Hl$, we have $S=\Hl$. By
Lemma~\ref{lemma:condition2} the theorem is true.
\end{proof}

\begin{theorem} \label{theorem:local-one-sink}
Let $G=(V,E)$ be a weakly connected digraph such that
$|S(\mathbb{G})|=|\{\mathbf s\}|=1$ and let $\LL= \{c_{(v,w)}\ketbra{w}{v}\colon
(v,w)\in E\}$ for some $c_{(v,w)}\in \C_{\neq 0}$. Then the evolution described by
Eq.~(\ref{eq:GKSL-master-equation}) is relaxing for arbitrary Hamiltonian $H$.
\end{theorem}
\begin{proof}
Suppose $S_1\neq \{0\},S_2\neq \{0\}$ are two  subspaces of $\Hl$ and let
$\ket{\psi_1}\in S_1,\ket{\psi_2}\in S_2$. Similarly to method in
Theorem~\ref{theorem:local-strongly} one can show that there exist
$L_1^1,\dots,L_k^1\in \LL$ and $L_1^2,\dots,L_{k'}^2\in \LL$ such that
$c_w^1\ket{w} = L_k^1L_{k-1}^1\dots L_1^1 \ket{\psi_1}$ and $c_w^2\ket{w} =
L_{k'}^2L_{k'-1}^2\dots L_1^2 \ket{\psi_2}$ for some $w\in \mathbf s$ and
$c_w^1,c_w^2\in \C_{\neq 0}$. Hence $S_1$ and $S_2$ are not orthogonal and by
Lemma~\ref{lemma:condition3} the theorem holds.
\end{proof}
Note that no information about the graph structure needs to be encoded in the
Hamiltonian---the graph structure is encoded in the $\LL$ only.

The remaining class of weakly connected graphs is those for which $|S(\mathbb
G)|>1$. However in this case one can shown that Hamiltonian has impact on limiting
behaviour of the evolution. In the case of $\omega=1$, one can show that the
evolution is equivalent to the classical one. Because of that, there are different
stationary states in each of the sinks. Hence, the evolution is not relaxing.
However, due to de-cohering character of the evolution, it should be
convergent to the state from the sinks subspace.

Similarly convergence is expected in the case of $\omega\in(0,1)$. In
this case, if the Hamiltonian impact is sufficiently large, one can observe the
relaxation of the evolution.  As an example, we analyse a bidirected path with
vertex set $\{-n,-(n-1),\dots,n-1,n\}$, and edge set $E=\{(i,i+1):i<0\} \cup
\{(i-1,i):i>0\}$ (see Fig.~\ref{fig:bidirected-path}). The Hamiltonian is chosen to be
the adjacency matrix of the underlying graph, \ie{} the path graph. In
Fig.~\ref{fig:bidirected-path-threshold} the values of $\omega$ for which the evolution is relaxing are presented. One can see that for each $n$
there is $\omega_t$ such that for $\omega < \omega_t$ evolution is relaxing, and
for $\omega>\omega_t$ is not. As the value of $\omega_t$ decreased with $n$, one can observe that for larger graphs the stronger influence of the coherent part is necessary to ensure the relaxing property.


\begin{figure}\centering
 \begin{tikzpicture}[node distance=2cm]
\tikzset{nodeStyle/.style = {circle,draw,minimum size=3.5em}}

\node[nodeStyle] (A)  {$-n$};
\node[] (C) [right of=A] {$\dots$};
\node[nodeStyle] (D) [right of=C] {$-1$};
\node[nodeStyle] (E) [right of=D] {$0$};
\node[nodeStyle] (F) [right of=E] {$1$};
\node[] (G) [right of=F] {$\dots$};
\node[nodeStyle] (I) [right of=G] {$n$};

\tikzset{EdgeStyle/.style   = {->,>=latex}}
\draw[EdgeStyle] (C) to (A);
\draw[EdgeStyle] (D) to (C);
\draw[EdgeStyle] (E) to (D);
\draw[EdgeStyle] (E) to (F);
\draw[EdgeStyle] (F) to (G);
\draw[EdgeStyle] (G) to (I);

\end{tikzpicture}
\caption{Biderected path graph of size $2n+1$.}\label{fig:bidirected-path}
\end{figure}
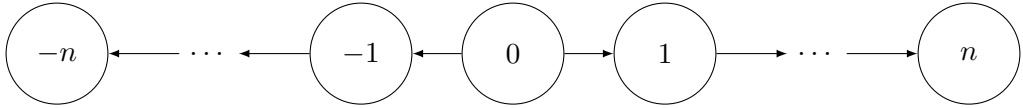

\begin{figure}\centering
	\includegraphics[width=0.7\textwidth]{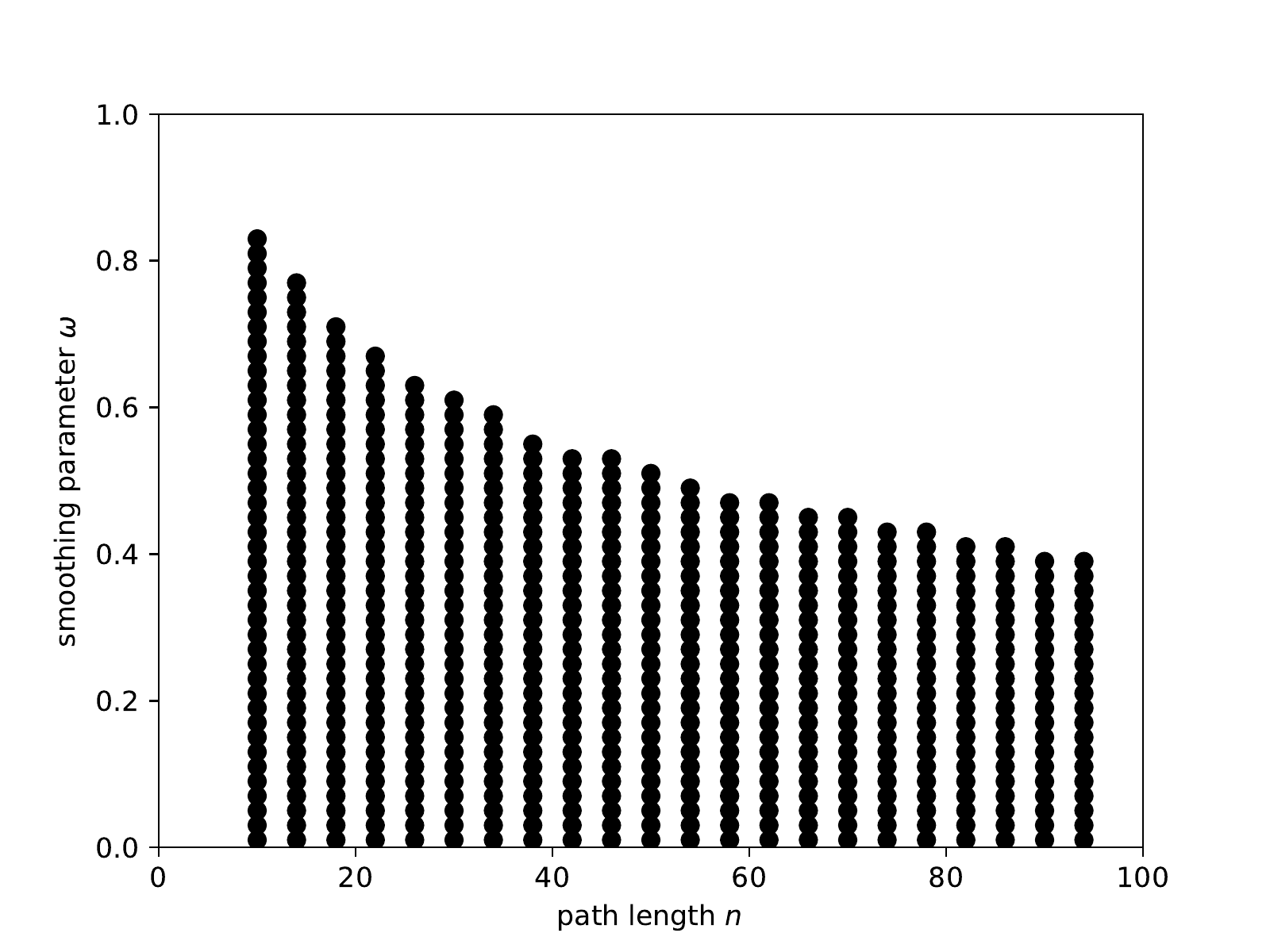}
	\caption{Values of $\omega$ for which the evolution on bidirected path graph is relaxing on local environment interaction evolution. Note that for each $n$ there is a threshold $\omega_t$, before which the evolution is relaxing. Numerical analysis was performed for $n=10,14,\dots,94$ and for $\omega=0.01,0.03,\dots,0.99$.
	}\label{fig:bidirected-path-threshold}
\end{figure}

We have analysed the graphs from Erd\H{o}s-R\'enyi model $\mathcal G(n,\frac{1}{10})$ with
$|S(\mathbb G)|>1$. For each size $n=10,15,\dots,45$ we have tested 200 graphs with values 
$\omega\in\{0.05,0.1,\dots,1\}$. We have ignored the graphs for which $|S(\mathbb G)|>1$ was not satisfied or which were not weakly connected. We have not found any example of graph for which the evolution was non-convergent. All graphs yield
convergent evolution for $\omega=1$.  Moreover, for all graphs it was possible to find such $\omega_t$, that for all $\omega\in(0,\omega_t)$ the evolution is relaxing, while for $(\omega_t,1]$ the evolution is convergent. 

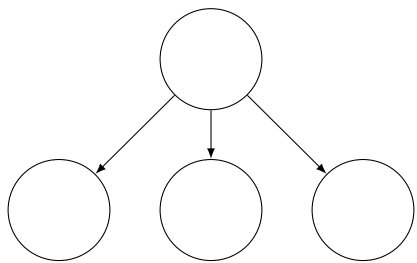
\begin{figure}
	\centering
 \begin{tikzpicture}[node distance=2cm]
\tikzset{nodeStyle/.style = {circle,draw,minimum size=3.5em}}

\node[nodeStyle] (A)  {};
\node[nodeStyle] (C) [below of=A] {};
\node[nodeStyle] (B) [left of=C] {};
\node[nodeStyle] (D) [right of=C] {};

\tikzset{EdgeStyle/.style   = {->,>=latex}}
\draw[EdgeStyle] (A) to (B);
\draw[EdgeStyle] (A) to (C);
\draw[EdgeStyle] (A) to (D);

\end{tikzpicture}
\caption{Star graph of size 4, which yields non-relaxing evolution for $\omega\in(0,1]$}\label{fig:multi-sink-nonrelaxing-classical}
\end{figure}

However, we were able to identify graphs which are non-relaxing for all $\omega\in(0,1)$. Let us consider a star graph in Fig.~\ref{fig:multi-sink-nonrelaxing-classical}. One can show, that for arbitrary $\omega\in(0,1]$ the evolution operator $F$ consists of at least 4 zero eigenvalues. Hence, the evolution for such graph is never relaxing. Star graphs of bigger size yield similar properties. On th other hand, we were not able to find any graph, which yields relaxing evolution for all $\omega\in(0,1)$.

\section{Convergence of global interaction case} \label{sec:global}

\subsection{Undirected graphs} We start this section with providing the general
result for the commuting operators. The result generalizes the case of quantum
stochastic walk with global interaction for all undirected graphs and for some
directed graphs~\cite{domino2016properties}, including circulant matrices.

\begin{proposition}\label{theorem:commuting}
Let us consider Eq.~\eqref{eq:evolution} in the case of commuting 
Lindbladian operators $\LL$ and Hamiltonian $H$. Then the evolution 
operation is of the form
\begin{equation}
(U\kron \bar U)\exp (tD_{F_{\omega}})(U\kron \bar U)^\dagger,
\end{equation}
where
\begin{equation}
D_{F_{\omega}}= -\ii (1-\omega) (D_H\kron \Id 
-\Id\kron 
 D_H) + 
\omega \sum_{L\in\LL}(D_L\kron \bar D_L - \frac{1}{2} \bar D_L D_L\kron 
\Id-\frac{1}{2} \Id 
\kron 
 D_L \bar  D_L)  \label{eq:diagonal-commuting}
\end{equation}
is a diagonal matrix. Here we assume that $U$ is a unitary operator and $D_H,D_L$ are
diagonal operators such that $H=UD_HU^\dagger$ and $L=UD_LU^\dagger$.
\end{proposition}
\begin{proof}
The proof comes directly from the eigendecompositions of the operators. Since
all operators commute, it is possible to find common eigendecomposition with the
same unitary matrix. By this we can easily find the result.
\end{proof}
One can note the global interaction quantum stochastic walk on undirected graphs
is a special case of the evolution described in the above theorem. In the walk
model the difference comes from the size of $\LL$, where we choose only single
Lindbladian operator. Hence we prove a result concerning undirected graphs.

\begin{theorem} \label{theorem:undirected}
The stationary states in the GKSL master equation evolution for which
$\LL=\{H\}$  are precisely the stationary states of the pure continuous quantum
evolution.  The evolution is convergent for $\omega\in(0,1]$, but
not relaxing iff the system size is greater than one.
\end{theorem}
\begin{proof} By the model construction  we have $\LL={H}$. Hence the 
formula Eq.~\eqref{eq:diagonal-commuting} simplifies to
\begin{equation}
D_{F_{\omega}}=-\ii (1-\omega) (D\kron \Id 
-\Id\kron 
 D) + 
\omega (D \kron D - \frac{1}{2} D^2\kron 
\Id-\frac{1}{2} \Id 
\kron 
 D ^2). 
\end{equation}
Here we assume $H=UDU^\dagger$. Since $H$ is hermitian, operator $D$ is a
real-valued diagonal matrix. The diagonal entries of operator $D_{F_{\omega}}$
are eigenvalues which characterize the evolution. We have
\begin{equation}
\begin{split}
\bra {i,j} D_{F_{\omega}}\ket {i,j} &=  -\ii(1-\omega) (\bra i D \ket i - \bra
j D\ket j)+\omega ( \bra i D \ket i \bra j D \ket j  -
\frac{1}{2}(\bra{i}D^2\ket i+\bra{j}D^2\ket j)) =\\
&= -\ii(1-\omega) (\bra i D
\ket i - \bra j D\ket j)- \frac\omega 2 (\bra i D \ket i - \bra j D\ket j)^2.
\end{split}
\end{equation}

Note $-\ii (1-\omega ) (\bra i D \ket i - \bra j D \ket j)$ corresponds to purely Hamiltonian evolution, and hence to continuous quantum walk. Since 0-eigenvalues of $F_\omega$ correspond to 0-eigenvalues of $H$ Hamiltonian part of the system, which furthermore correspond to the stationary states of the continuous-time quantum walk, we obtained the first part of the theorem.

Note that there are no purely imaginary eigenvalues od $F_\omega$. Hence, we have that the evolution is convergent. Since the set of stationary states of continuous-time quantum walk of size $n$ has at least $n$ elements, we obtain that QSW with global interaction is never relaxing.
\end{proof}
Note that the result from the above theorem implies that we can generate the
stationary states from the continuous quantum walk by adding the same
Lindbladian operator. 

\begin{remark}
	Global interaction case QSW is convergent, but not relaxing for arbitrary undirected graph with number of vertices greater than one 1 and for arbitrary $\omega \in (0,1]$. Furthermore, the stationary states are precisely those from CQW.
\end{remark}

In the next section we show that Theorem~\ref{theorem:undirected} cannot be
generalized for directed graphs.

\subsection{Directed graphs} In this section we
provide an example of a directed graph for which we do not necessary obtain
a stationary state. It has been proven that the evolution converges for arbitrary
initial state iff all nonzero eigenvalues of $F_{\omega}$ have negative real
part \cite[Theorem 5.4]{rivas2012open}. We found an example of a digraph which do not satisfy
the condition, and provide an exemplary initial state which results in
non-convergent evolution. One should note, that it is possible (and more probable) 
to find non-convergent states.

\begin{theorem}
Let us take the evolution for which the only Lindbladian operator is an
adjacency matrix of the directed graph and the Hamiltonian is an adjacency
matrix of the underlying graph. Then there exists a directed graph $G$ and an
initial state $\varrho_0$ for which the evolution is non-convergent for an
arbitrary value of the smoothing parameter $\omega\in(0,1]$.
\end{theorem}
\begin{proof}
\begin{figure}
\centering
\begin{tikzpicture}

\def \n {8}
\def \radius {2.5cm}
\def \margin {8} 

\foreach \s in {7,...,0}
{
	\node[draw, circle] (\s) at ({360/\n * (\s - 1)}:\radius) {$\s$};
	\draw[<->, >=latex] ({360/\n * (\s - 1)+\margin}:\radius)
	arc ({360/\n * (\s - 1)+\margin}:{360/\n * (\s)-\margin}:\radius);
}
\draw [<-,>=latex] (0)  edge (2);
\draw [<-,>=latex] (1)  edge (3);
\draw [<-,>=latex] (2)  edge (4);
\draw [<-,>=latex] (3)  edge (5);
\draw [<-,>=latex] (4)  edge (6);
\draw [<-,>=latex] (5)  edge (7);
\draw [<-,>=latex] (6)  edge (0);
\draw [<-,>=latex] (7)  edge (1);

\end{tikzpicture} 
\caption{An example of strongly connected directed graph, for which the global interaction case evolution is not convergent.}\label{fig:example-small-directed}
\end{figure}
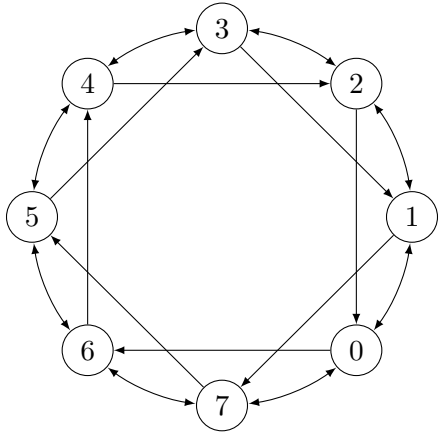
As an example we choose a circulant graph of size $4k$ for $k>1$ and with extra jump every two vertices. An example for $k=2$ is presented in
Fig.~\ref{fig:example-small-directed}. The graph and its underlying graph are
circulant matrices. Therefore, we can use Eq.~\eqref{eq:diagonal-commuting}
to find out that there exists one eigenvalue of the form $2(1-\omega)\ii$ with
corresponding eigenvector $\ket{C_{k}}\overline {\ket{C_{2k}}}$, where
$\overline{(\cdot)}$ denotes the element-wise conjugation and $\ket {C_i}$  is
the $i$-th eigenvector of a circulant matrix of the form
\begin{equation}
\ket {C_i}= \frac{1}{2\sqrt k}\sum_{j=0}^{4k} \exp \left(\frac{2\pi \ii 
ij}{4k-1}\right) \ket i.
\end{equation} 
We need to find a matrix which is not orthogonal to 
$\ketbra{C_k}{C_{2k}}$. Our exemplary initial state is 
\begin{equation}
\begin{split}
\varrho_0  &= \frac{1}{2}(\ket{C_k}+\ket{C_{2k}})(\bra{C_k}+\bra{C_{2k}}).
\end{split}
\end{equation}
The $\varrho_t$ takes the form
\begin{equation}
\varrho_t = \frac{1}{2}(\ketbra{C_k}{C_k} + \ketbra{C_{2k}}{C_{2k}} + 
e^{2\ii(1-\omega)t}\ketbra{C_k}{C_{2k}}+ e^{-2\ii(1-\omega)t}\ketbra{C_{2k}}{C_k}).
\end{equation}
Since $\varrho_t$ is periodic with period $\frac{\pi}{(1-\omega)}$, we obtain 
the result.
\end{proof}
Note, that for different $t$ we can obtain different state in the sense of
possible measurement output. For example we have $\bra 0 \varrho_0\ket 0 =
\frac{1}{2k}$, but at the same time we have $\bra 0 \varrho
(\frac{\pi}{2(1-\omega)}) \ket 0 = 0$. 

\begin{remark}
The evolution for which the only Lindbladian operator is an adjacency matrix of
the directed graph and the Hamiltonian is an adjacency matrix of the underlying
graph, does not converge in general, even in the sense of the canonical
measurement probability distribution.
\end{remark}

Circulant graphs provide an infinite collection of directed graphs for which the convergence does not hold. Note that the example used in the proof is strongly connected directed graph. This shows, that the convergence in the local interaction case does not imply the convergence in the global interaction case. 

Contrary, it is very difficult to find a directed graph for which the convergence does not hold. We have made numerical analysis for Erd\H{o}s-R\'enyi model $\mathcal{G}(n,0.5)$ for $n=10,15,\dots,45$ and for $\omega=0.05,0.1,\dots,1$. We have guaranteed that graph is weakly connected by ignoring such cases if they appear. For each $n$ we chose 200 graphs and we have not found any graph which is not convergent. At the same time it is possible to find relaxing evolution. We were not able to identify graph properties causing such behaviour. From all 1600 graphs, 1599 were relaxing for all value of $\omega$, and there was one graph of size 10 which was relaxing for $\omega\in\{0.05,\dots,0.95\}$ and convergent for $\omega= 1$. Therefore, relaxing property is statistically common for directed graphs in the global environment interaction case.

\section{Convergence of non-moralizing global interaction case}\label{sec:corrected-global}
\subsection{Non-convergent in the space state}

The non-moralizing model has been introduced in \cite{domino2017spontaneous} as a
provide an similar example of a graph and initial state such that it does not
converge. Similarly we found a digraph, for which an evolution operator
$S_{t,\omega}$ has an imaginary part. Again, it is easy to find a state for 
which
the evolution is non-convergent.

\begin{theorem} \label{th:global-nonmoralizing}
Let us take the non-moralizing evolution described in
\cite{domino2017spontaneous}. Then there exists a directed graph $G$ and initial
state $\varrho_0$ for which the evolution is periodic in time for an arbitrary
value of the smoothing parameter $\omega\in(0,1]$.
\end{theorem}

\begin{proof}
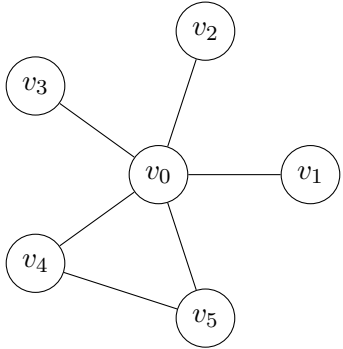
\begin{figure}
\centering
\begin{tikzpicture}
\node[draw,circle] (0) {$v_0$};
\def \n {5}
\def \radius {2cm}
\def \margin {8} 

\def \s {1}
\node[draw, circle] (\s) at ({360/\n * (\s - 1)}:\radius) {$v_\s$};
\def \s {2}
\node[draw, circle] (\s) at ({360/\n * (\s - 1)}:\radius) {$v_\s$};
\def \s {3}
\node[draw, circle] (\s) at ({360/\n * (\s - 1)}:\radius) {$v_3$};
\def \s {4}
\node[draw, circle] (\s) at ({360/\n * (\s - 1)}:\radius) {$v_4$};
\def \s {5}
\node[draw, circle] (\s) at ({360/\n * (\s - 1)}:\radius) {$v_\s$};

\draw [-,>=latex] (0)  edge (1);
\draw [-,>=latex] (0)  edge (2);
\draw [-,>=latex] (0)  edge (3);
\draw [-,>=latex] (0)  edge (4);
\draw [-,>=latex] (0)  edge (5);
\draw [-,>=latex] (5)  edge (4);
\end{tikzpicture} 
\caption{An example of graph for which non-moralizing global interaction evolution is not convergent.}\label{fig:example-big-directed}
\end{figure}

Let us take a graph presented in Fig.~\ref{fig:example-big-directed}. Using the
scheme presented in \cite{domino2017spontaneous}, new graph will consist of 5
copies of vertex $v_0$, two copies of vertices $v_4$ and $v_5$, and single copy of
other vertices. As the orthogonal matrices we choose the Fourier matrices and 
in order to remove the premature localization  we choose the rotating Hamiltonian $\tilde
H_{\textrm{rot}} $ constructed from the Hamiltonians of the form
\begin{equation}
 \begin{bmatrix}
0 & \ii &  &  \\
-\ii & \ddots & \ddots   & \\
& \ddots & \ddots & \ii & \\
&  & -\ii & 0 & \\
\end{bmatrix}.
\end{equation}

Let us choose two eigenvectors of the rotating Hamiltonian
\begin{gather}
\ket{\lambda ^1_{\tilde H_{\textrm{rot}}}} = \frac{1}{2\sqrt 3}\ket{v_0^0} 
-\frac{\ii}{2}\ket{v_0^1}- \frac{1}{\sqrt 3}\ket{v_0^2} + \frac{\ii}{2} 
\ket{v_0^3}+\frac{1}{2\sqrt 3}\ket{v_0^4},\\
\ket{\lambda^2_{\tilde H_{\textrm{rot}}}} = \frac{1}{2\sqrt 3}\ket{v_0^0} 
-\frac{\ii}{2}\ket{v_0^1}- \frac{1}{\sqrt 3}\ket{v_0^2} + \frac{\ii}{2} 
\ket{v_0^3}+\frac{1}{2\sqrt 3}\ket{v_0^4}.\\
\end{gather}
One can show that the vectors $\ket{\lambda ^1_{\tilde 
H_{\textrm{rot}}},\bar\lambda ^1_{\tilde H_{\textrm{rot}}}},\ket{\lambda 
^1_{\tilde 
H_{\textrm{rot}}},\bar\lambda ^2_{\tilde H_{\textrm{rot}}}},\ket{\lambda 
^2_{\tilde 
H_{\textrm{rot}}},\bar\lambda ^1_{\tilde H_{\textrm{rot}}}},\ket{\lambda 
^2_{\tilde 
H_{\textrm{rot}}},\bar\lambda ^2_{\tilde H_{\textrm{rot}}}}$ are the 
eigenvectors of the increased evolution operator $\tilde S_{t,\omega}$ for 
arbitrary 
$\omega\in(0,1]$. Corresponding eigenvalues are respectively $0,-2\ii\sqrt 
3\omega,2\ii\sqrt 3\omega,0 $. Similarly to the example presented in the 
previous 
section, the 
state
\begin{equation}
\tilde \varrho_0 = \frac{1}{2}(\ket{\lambda ^1_{\tilde 
H_{\textrm{rot}}}}+\ket{\lambda ^2_{\tilde 
H_{\textrm{rot}}}})(\bra{\lambda ^1_{\tilde 
H_{\textrm{rot}}}}+\bra{\lambda ^2_{\tilde 
H_{\textrm{rot}}}})
\end{equation}
is the required initial state. The state after time $t$ takes the form
\begin{equation}
\tilde \varrho_t = \frac{1}{2}(\ket{\lambda ^1_{\tilde 
H_{\textrm{rot}}}}\bra{\lambda ^1_{\tilde 
H_{\textrm{rot}}}}+e^{-2\ii t\sqrt{3}\omega}\ket{\lambda ^1_{\tilde 
H_{\textrm{rot}}}}\bra{\lambda ^2_{\tilde 
H_{\textrm{rot}}}}+e^{2\ii t\sqrt{3}\omega}\ket{\lambda ^2_{\tilde 
H_{\textrm{rot}}}}\bra{\lambda ^1_{\tilde 
H_{\textrm{rot}}}}+\ket{\lambda ^2_{\tilde 
H_{\textrm{rot}}}}\bra{\lambda ^2_{\tilde 
H_{\textrm{rot}}}}).
\end{equation}
The function $\tilde\varrho_t$ is periodic with period 
$\frac{\pi}{\sqrt3\omega}$, hence we obtained the result.
\end{proof}

\subsection{Convergence in the sense of the measurement}

The example from Theorem~\ref{th:global-nonmoralizing}, the probability
distribution obtained from the measurement in the canonical basis of the system
changes in time. However, when we analyse the canonical measurement from
\cite{domino2017spontaneous}, where as the measurement operators we choose the
projections onto the subspaces corresponding to different vertices, we can
observe that the probability does not change---in this case the probability
of measuring  vertex $v_0$ for each time point is one. 

We have performed numerical analysis for graphs of size $n=9$ from Erd\H{o}s-R\'enyi $\mathcal G(n,0.5)$ model for $\omega\in\{0,0.1,\ldots,1\}$. We have checked 200 random graphs and the non-moralizing QSW on each of them was convergent.

\begin{conjecture}
Let us choose the non-moralizing evolution model. Let $\Pi(\varrho_0,t)$ 
denotes 
the probability distribution of canonical 
measurement onto the subspaces of 
vertices in time $t$ with the initial state $\varrho_0$. Then for arbitrary 
$\varrho_0$ there exists probability distribution $\Pi_\infty$ such that 
\begin{equation}
\lim_{t\to\infty}\Pi(\varrho_0,t) = \Pi_\infty.
\end{equation}
\end{conjecture}

Note, that the probability distribution may be nonunique. To see that let us
analyse the graph presented in
Fig.~\ref{fig:big-evolution-different-measurements}. We choose
$\omega=\frac{1}{2}$ and two initial states
$\tilde\varrho_0=\ketbra{v_6^0}{v_6^0}$ and
$\tilde\varrho_1=\ketbra{v_7^0}{v_7^0}$. We have found the limiting probability
distribution
\begin{gather}
\Pi_0 \coloneqq\lim\limits_{t\to\infty}\Pi(\varrho_0,t), \\
\Pi_1 \coloneqq\lim\limits_{t\to\infty}\Pi(\varrho_1,t).
\end{gather}
The probability distributions differs, for example $\Pi_0(v_0) = 0.666616$ and 
$\Pi_1(v_0)=0.11897$.

\begin{figure}
\centering

 \begin{tikzpicture}[node distance=2cm]
  \tikzset{nodeStyle/.style = {circle,draw,minimum size=2.5em}}
  
  \node[nodeStyle] (A)  {$v_7$};
  \node[nodeStyle] (D) [below right of=A] {$v_2$};
  \node[nodeStyle] (B) [above right of=D] {$v_4$};
  \node[nodeStyle] (E) [below right of=B] {$v_1$};
  \node[nodeStyle] (C) [above right of=E] {$v_5$};
  \node[nodeStyle] (F) [below of=D] {$v_6$};
  \node[nodeStyle] (G) [below of=E] {$v_3$};

  \tikzset{EdgeStyle/.style   = {}}
  \draw[EdgeStyle] (A) to (D);
  \draw[EdgeStyle] (D) to (B);
  \draw[EdgeStyle] (B) to (E);
  \draw[EdgeStyle] (C) to (E);
  \draw[EdgeStyle] (D) to (F);
  \draw[EdgeStyle] (E) to (G);
  \draw[EdgeStyle] (D) to (E);
  \draw[EdgeStyle] (F) to (E);
  \draw[EdgeStyle] (D) to (G);
  
  \end{tikzpicture}

\caption{Graph for which there exists two different stationary states in the 
sense of the natural measurement. The states can be obtained by starting in 
vertices $v_7$ and $v_6$.}\label{fig:big-evolution-different-measurements}
\end{figure}
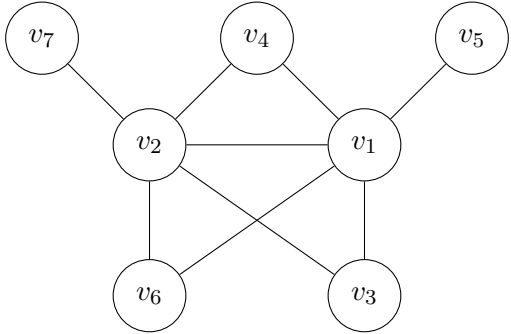

\subsection{The digraph structure observance} 
\label{sec:digraph-structure-observance}

In \cite{domino2017spontaneous} it was
suggested, that the non-moralizing global environment interaction QSW can be
applied for modelling the walk on directed graphs. Moreover, an example
suggesting that the original global interaction evolution, where Lindbladian
operator is an adjacency matrix of the directed graph, does not preserve the
digraph structure have been provided. Let us take
the graph from Fig.~\ref{fig:intuitive-graph} and let us analyse the graph
without the Hamiltonian. One can find that
\begin{equation}
\frac{1}{4}(\ket{v_1}-\ket{v_2})(\bra{v_1}-\bra{v_2}) + 
\frac{1}{2}\ketbra{v_3}{v_3}\label{eq:unintuitive}
\end{equation}
is proper stationary state. There is a nonzero probability of measuring the
state in vertex $v_1$ and $v_2$. Thus, one can see that the superoperator constructed
using a directed graphs can lead to the evolution on some other structure, namely the
directed moral graphs. This behaviour, expected in the case when the Hamiltonian part is present,
is surprising in the pure Lindbladian case.  

For the purpose of quantifying this behaviour we introduce the following property
of the evolution on directed graphs.

\begin{property}[Digraph structure observance]
	Let us assume that for an arbitrary vertex there is path to
	some sink vertex. We say that the evolution on a directed graphs 
	has \emph{digraph structure observance} property if an arbitrary 
	initial state converges to the state spanned by vectors corresponding
	to the sink vertices from the condensation of the graph.
\end{property}

We require that the evolution on directed graph, at least for Lindbladian part,
should have the digraph structure observance property.

One should note that digraph structure observance can be quantified for any directed
graph using additional sink vertices attached to the orginal vertices. 

In the case of the example in Fig.~\ref{fig:lindbladWalkGraphExample}
arbitrary state should converge to the state $\ketbra{v_3}{v_3}$.

\begin{figure}
\centering
\subfigure[\label{fig:intuitive-graph}]{\begin{minipage}[b]{.35\linewidth}
\centering
 \begin{tikzpicture}[node distance=2cm]
  \tikzset{nodeStyle/.style = {circle,draw,minimum size=2.5em}}
  
  \node[nodeStyle] (A)  {$v_1$};
  \node[nodeStyle] (C) [below right of=A] {$v_3$};
  \node[nodeStyle] (B) [above right of=C] {$v_2$};
  
  \tikzset{EdgeStyle/.style   = {->,>=latex}}
  \draw[EdgeStyle] (A) to (C);
  \draw[EdgeStyle] (B) to (C);
  
  \end{tikzpicture}
  
  \end{minipage}}
 \hspace{1cm}
\subfigure[\label{fig:true-graph}]{\begin{minipage}[b]{.35\linewidth}
\centering
 \begin{tikzpicture}[node distance=2cm]
  \tikzset{nodeStyle/.style = {circle,draw,minimum size=2.5em}}
  
  \node[nodeStyle] (A)  {$v_1$};
  \node[nodeStyle] (C) [below right of=A] {$v_3$};
  \node[nodeStyle] (B) [above right of=C] {$v_2$};
  
  \tikzset{EdgeStyle/.style   = {->,>=latex}}
  \draw[EdgeStyle] (A) to (C);
  \draw[EdgeStyle] (B) to (C);
  
  \tikzset{additionalEdgeStyle/.style = {dashed}}
  \draw[additionalEdgeStyle] (A) -- (B);
  \end{tikzpicture}

  \end{minipage}}
  
    \caption{Visualisation of a directed graph \subref{fig:intuitive-graph} 
    and 
    its spontaneous moralization
    \subref{fig:true-graph}.  
    }\label{fig:lindbladWalkGraphExample}
\end{figure}
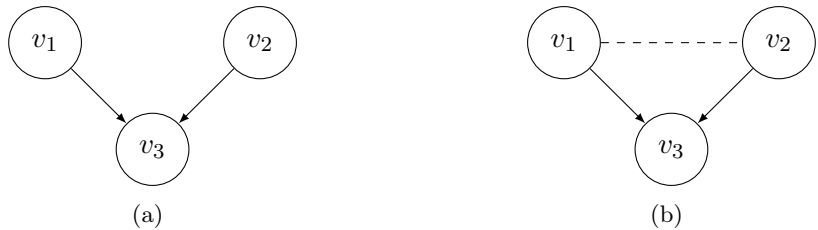

The unintuitive stationary state in Eq.~\eqref{eq:unintuitive} comes from the
spontaneous moralization of the graph \cite{domino2017spontaneous}. Since the
moralization of the graph was corrected, it is necessary to check whether the
improper stationary state still occurs in the non-moralizing evolution. This 
can be achieved by analysing the stationary states.

The numerical analysis was performed as follows. Let $S$ be a set of all sink
vertices. We start in some vertex with nonzero outdegree. We determine the state
$\varrho_\infty$ for large time value and we verified numerically, whether it is
close to stationary state in the sense of probability distribution of
measurement. Then we compute the cumulated probability of measuring the state in
the sink vertices
\begin{equation}
p_{S}(\varrho_\infty) = \sum_{v\in \bigcup S(\mathbb G)} \bra{v}\varrho_\infty\ket v
\end{equation}
and the second moment of the distance from the sink vertices
\begin{equation}
\mu_{S}(\varrho_\infty) = \sum_{v\in V} d^2(v,S) \bra v 
\varrho_\infty \ket v,
\end{equation}
where $d(v,S)$ is the length of the shortest path from $v$ to closest sink
vertex, \ie{} $d(v,S)=\min_{s\in \bigcup S(\mathbb G)}d(v,s)$. We say, that the
greater the value of $p_{S}$ and the lower the value of $\mu_{S}$ are, the more
the evolution preserves the graph.

We have analysed the evolution for
$\omega\in(0,1]$. For the purpose of our analysis we have selected four types of
graphs, namely: a path graph, the Petersen graph, the Apollonian graph, and the
Sierpinski triangle. We choose the orientation of the graphs such that each
vertex is either a sink vertex, or there is a path from it to some sink vertex.

The obtained results are presented in Fig.~\ref{fig:model-check}. One can see
that the larger the value of $\omega$ is, the more probability cumulates in the
close neighbourhood of the sink vertices. In these examples for
$\omega\in[0.7,1]$,  $\mu_\textrm{S}$ and $p_\textrm{S}$ respectively decreases
and grows with $\omega$. One can also notice that in the $\omega\to1$ limit the
$p_{\textrm{S}}$ converge to one and $\mu_\textrm{S}$ vanishes. The above
analysis suggests that when $\omega=1$, the directed graph structure is fully
preserved.

We have done further analysis for Erd\H{o}s-R\'enyi graphs $\mathcal G(n,0.2)$ with $n=9,11,12$, for each size 500 samples. In all of
these graphs for $\omega=1$ we have $p_S=1$ and $\mu_S=0$, which suggests that
at least for this extreme value of $\omega$ the digraph structure is observed.

Furthermore, we have searched for the $\omega_0$, for which for all
$\omega>\omega_0$ both measures $p_S$ and $\mu_S$ were monotonic in $\omega$. We were selecting random graphs $\mathcal G(n,0.2)$ with $n=9,11,12$ and next we were calculating $p_S$ and $\mu_S$ for $\omega$ from 1 to 0 with step $-0.02$. When $p_S$ started increasing or $\mu_S$ started decreasing we stopped the calculations and save $\omega_0$.
The statistics of the threshold $\omega_0$ is presented in
Fig.~\ref{fig:omega-treshold}. We have found no graph, for which $\omega_0>0.7$. This may suggest that for $\omega>0.7$ structure of the directed graph is preserved at least for given sizes of the digraphs.However, to provide more information about this behaviour the more detailed analysis is necessary.

\begin{figure}[h!]
	\includegraphics[scale=1]{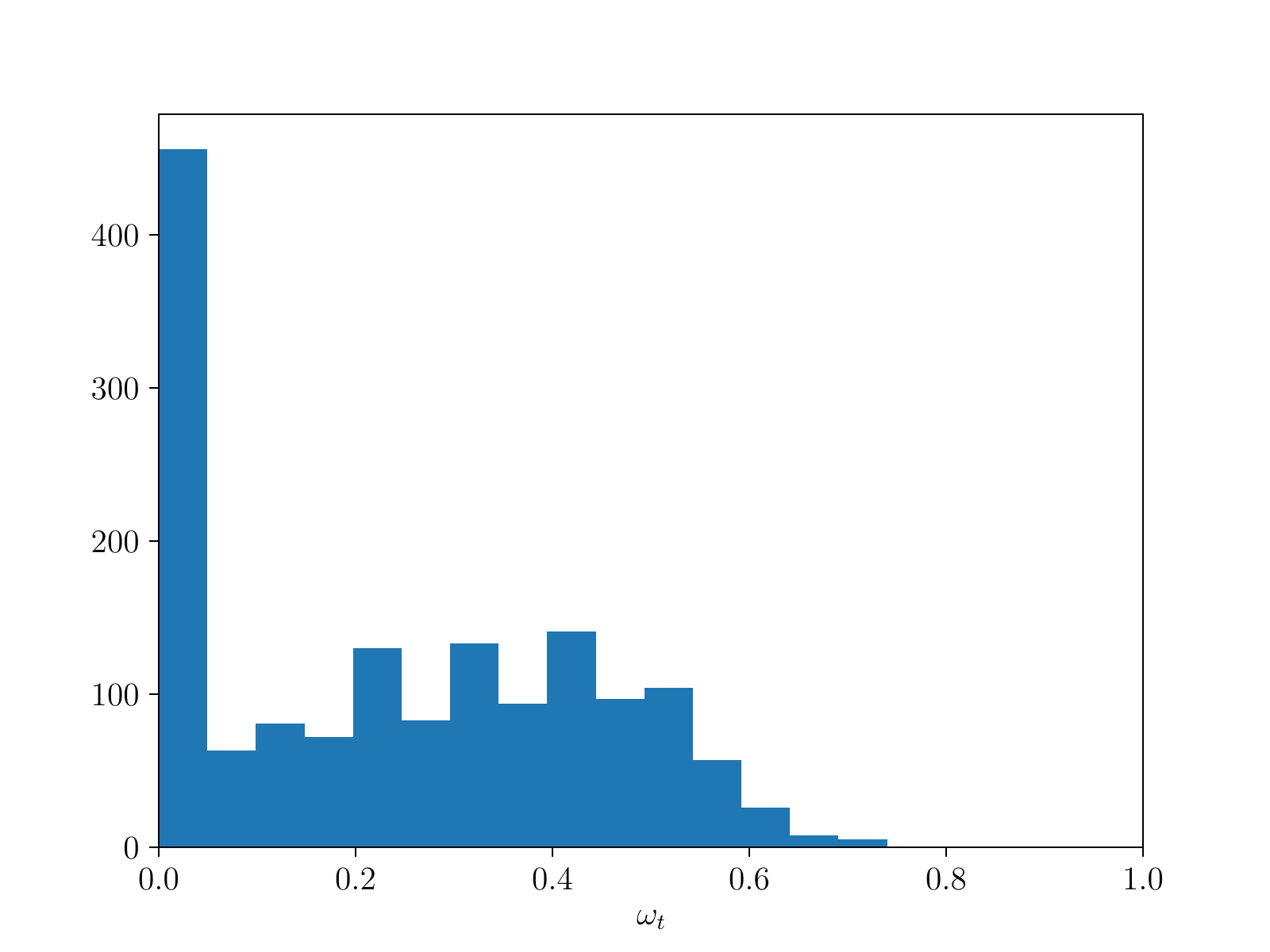}
	\caption{Histogram presents amount of graphs $\mathcal G(n,0.2)$ with $n=9,11,12$ having particular $\omega_0$.}
	\label{fig:omega-treshold}
\end{figure}


\begin{figure}
\centering
\includegraphics[width=0.8\textwidth]{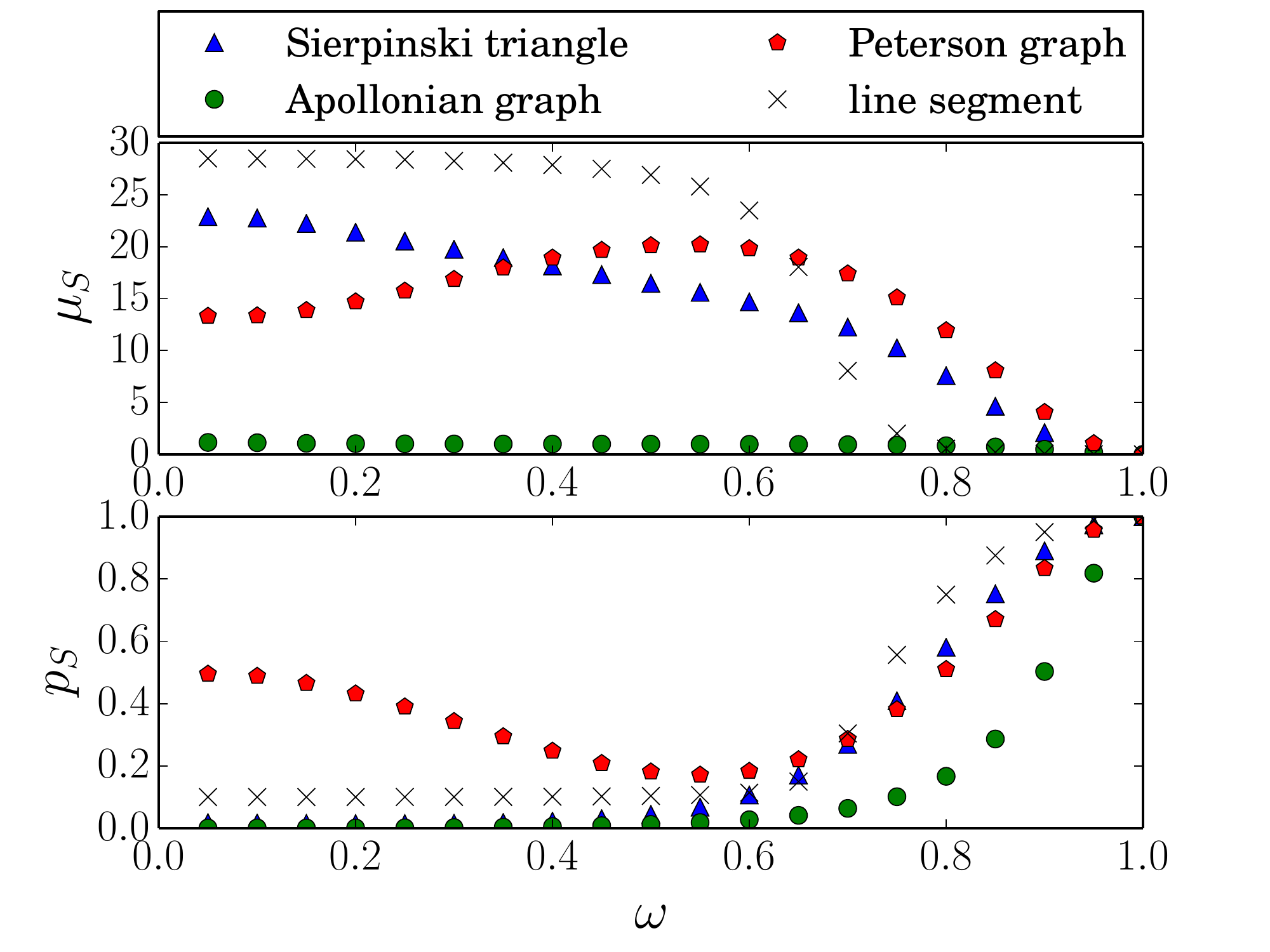}

    \caption{Top plot shows the second moment of the distance from the vertex to
	the closest sink vertex. Bottom plot shows the probability of measuring the
	only single sink vertex for large time graphs. Sierpinski triangle has 15
	vertices, Apollonian graph has 12 vertices and length of path graph equals
	10.
}\label{fig:model-check}
\end{figure}

%
%
%
\section{Concluding remarks}\label{sec:conclusions} In this article we analysed three cases of QSW:
with local interaction, with global interaction and with non-moralizing global
interaction. As for local interaction case our results generalize the results
from \cite{liu2016continuous}. We prove that the evolution is relaxing
in the case of undirected graphs. We also show that the result can be extended to
arbitrary strongly connected digraph and for weakly connected digraphs with one
sink in the condensation graph. Furthermore, we show that for strongly connected digraphs the
stationary state is located in the interior. At the same time we provide a counterexample demonstrating that the result cannot be extended to all directed graphs.

Global interaction case QSW is convergent, but never relaxing for arbitrary
undirected graph with number of vertices greater than one 1 and for arbitrary
$\omega \in (0,1]$. Furthermore the stationary states are precisely does from
CQW. We show by example, that the result cannot be extended into arbitrary
directed graphs. Surprisingly, applying the Hamiltonian may help relax the
evolution.

In the non-moralizing global interaction case we provide examples demonstrating
that the evolution does not need to be relaxing, or even convergent, even for undirected graphs.
We also  give an example of evolution, for which the evolution is not relaxing,
even in the sense of canonical measurement. However, we conjecture by numerical
analysis, that the evolution is convergent in the sense of the canonical measurement.

For the purpose of analysing digraph structure observance we have introduced sink observance property. This property can be quantified by
analysing the probability of measuring the state in the sink vertex and its
neighbourhood. We argue that the bigger the probability
of measuring the state in sink vertex or its closes neighbourhood, the better
the structure observance. 
Since the Lindblad operators corresponds to the directed graph structure, we
expect, that for $\omega$ close to $1$ the probability of measuring the sink
vertex is one. For the global interaction case the digraph structure is not
preserved, even for a very simple example. Fortunately, the digraph observance is
recovered in the non-moralizing global interaction case.

\section*{Acknowledgements} This work has been supported by the Polish National Science
Centre under project number 2011/03/D/ST6/00413. 
\bibliographystyle{ieeetr}
\bibliography{stochastic-stationary}

\end{document}